\documentclass[conference]{IEEEtran}

\usepackage[colorlinks]{hyperref}
\usepackage{cite}
\usepackage{float}
\usepackage{amsmath,amssymb,amsfonts,amsthm,nccmath}
\usepackage{physics}
\usepackage{bm}
\usepackage[algo2e]{algorithm2e}
\usepackage[dvipsnames]{xcolor}
\usepackage{subcaption}
\usepackage{multirow}
\usepackage{tabularx,threeparttable}
\usepackage{booktabs,tikz,soul,pifont}
\usepackage{enumitem}
\usepackage{graphicx}
\usepackage[most]{tcolorbox}

\hypersetup{
    linkcolor=blue,
    citecolor=blue,
    urlcolor=blue,
}

\newcommand{\rom}[1]{\expandafter{\romannumeral #1\relax}}

\newtheorem{definition}{Definition}[]

\newtheorem{theorem}{Theorem}[section]

\newtheorem{corollary}[theorem]{Corollary}

\newtheorem{observation}{Observation}[]

\DeclareMathOperator*{\argmin}{argmin}
\DeclareMathOperator{\TV}{TV}

\newcommand{\Chal}{\mathbf{Chal}}
\newcommand{\Adv}{\mathbf{Adv}}

\SetCommentSty{mycommfont}
\RestyleAlgo{ruled}
\newenvironment{game}[1][htb]
{
    
    \begin{algorithm2e}[#1]%
    \DontPrintSemicolon
    \LinesNumbered
    \SetNoFillComment
    \SetInd{0.4em}{0.4em}
    \setlength{\algomargin}{1.5em}
    \linespread{1.2}\selectfont
}{\end{algorithm2e}}

\newtcolorbox{tbox}[1][]{%
    colback=black!5,
    colframe=black!5,
    notitle,
    sharp corners,
    borderline west={1pt}{0pt}{blue},
    enhanced,
    breakable,
    left=0pt,
    right=0pt,
    top=0pt,
    bottom=0pt
}

\makeatletter
\renewcommand\paragraph{%
  \@startsection{paragraph}{4}%
    {\z@}%
    {1ex \@plus .5ex \@minus .2ex}%
    {0em}%
    {\normalfont\normalsize\bfseries}}
\makeatother

\title{Characterizing Privacy Risks of Quantum Machine Learning with Emergent Quantum-Native Access}
\author{
\IEEEauthorblockN{
    \textbf{Liou Tang}\textsuperscript{1},
    \textbf{James Joshi}\textsuperscript{1},
    \textbf{Ashish Kundu}\textsuperscript{2}
}
\IEEEauthorblockA{
    \textsuperscript{1}\textit{University of Pittsburgh, Pittsburgh, PA, USA} \\
    \textsuperscript{2}\textit{Cisco Research, San Jose, CA, USA}
}
liou.tang@pitt.edu \quad jjoshi@pitt.edu \quad ashkundu@cisco.com
}

\begin{document}
\maketitle
\thispagestyle{empty}
\pagestyle{empty}

\begin{abstract}

Quantum Machine Learning (QML) has shown rapid advances by utilizing quantum computing for machine learning tasks. Meanwhile, the privacy risks accompanying QML is also starting to be studied, which inherit privacy leakage channels from ``classical'' ML and also quantum-unique risks. Existing work on privacy-preserving QML largely focuses on a QML-as-a-service scenario, which generally assumes that the QML model owner provides only classical bit outputs to queries, while users (and adversaries) have only classical computing abilities. However, this view is increasingly challenged in a quantum-native world of \emph{quantum-capable users/adversaries}, which may have access to both quantum computing abilities and access to quantum information output from service providers.

In this paper, we aim to bridge this gap by examining membership inference attacks against QML models by demonstrating that increasing quantum access and quantum computing abilities provides provable theoretical privacy leakage and empirical adversarial gain. However, the probabilistic nature of QML introduces a gap between theoretical and empirical adversarial advantage. These results show that existing research on privacy leakage in QML models underestimates privacy leakage in emergent quantum-native access regimes, and we hope to establish a first step in examining potential privacy leakages for QML in the quantum-native world.

\end{abstract}

\section{Introduction}\label{sec:intro}

Quantum Computing (QC), which aims to leverage quantum mechanics (e.g., superposition, quantum entanglement), has advanced rapidly since its proposal \cite{Preskill2018QC}. The emerging body of work in Quantum Machine Learning (QML) \cite{Biamonte2017QML,Schuld2019QML,Cerezo2022Challenges} aims to integrate QC with classical Machine Learning (ML) tasks so as to allow representation of complex data in high-dimensional Hilbert space and reduce computational complexity for training QML models.

Yet, the switch from traditional to quantum computing is not frictionless. Existing work has shown that QML models remain vulnerable to various ``classical'' security and privacy attacks against ML, e.g., evasion/adversarial attacks \cite{Liao2021Robust,West2023Benchmark,Nowmi2026SoK}, backdoor \cite{Chu2023Qtrojan,Chu2023QDoor}, data poisoning \cite{Kundu2024Poison}, membership inference \cite{Heredge2025Character,Su2026MIAQMU}, data reconstruction \cite{Heredge2025Character}, and model stealing \cite{Fu2024QuantumLeak,Kundu2024Steal}. Further, QML models are also susceptible to security attacks that exploit quantum mechanisms, e.g., by manipulating quantum noise in QC devices \cite{Ash-Saki2021Noise,Ash-Saki2022Shuttle}.

Privacy risks of QC and QML have received less research attention than security risks until more recently \cite{Heredge2025Character,Su2026MIAQMU}.
A common assumption in the literature is a \emph{classical} adversary that queries the QML model through an API. The model outputs a classical measurement outcome, e.g., a prediction label or logits \cite{Su2026MIAQMU}. However, we argue that this view is insufficient to address the advancement of quantum computing devices \cite{Caleffi2024Survey} and emergent \emph{quantum networking}-based services \cite{Kimble2008Quantum,Wehner2018Quantum} and \emph{distributed quantum computing} \cite{Caleffi2024Survey,Kim2024Million}, in which the user/adversary can potentially submit and receive \emph{quantum information} \cite{Wehner2018Quantum,Devitt2016Cloud,Fu2024QuantumLeak}. These new access regimes could lead to additional privacy leakage in the model output that measurement alone does not capture, as they discard coherence and complementary-basis information \cite{Zurek2003Deco}. We therefore ask the following research questions:
\begin{quote}
\begin{enumerate}
    \item[\textbf{RQ1}:] Does an adversary with quantum computing abilities hold quantum advantages over a classical computing-only adversary in membership inference attacks?

    \item[\textbf{RQ2}:] Does access to quantum information allow the adversary to develop stronger membership inference attacks against QML models?
\end{enumerate}
\end{quote}
We aim to understand how quantum-capable adversaries, or quantum access to the target model, can increase model susceptibility to privacy leakage, particularly against \emph{membership inference attacks} (MIAs) \cite{Shokri2017MIA,Carlini2022MIA,Heredge2025Character,Su2026MIAQMU}. Our contributions are as follows:
\begin{itemize}
    \item We study MIAs against QML models under three access regimes: when the user/adversary only has access to classical outcomes ($\mathsf{C}$), can submit programmable measurements to the model owner/service providers ($\mathsf{M}$), and can receive pre-measurement quantum states ($\mathsf{Q}$).
    \item We formally establish a provable membership inference advantage for $\mathsf{C} \preceq \mathsf{M} \preceq \mathsf{Q}$, when more quantum access allows more privacy leakage (\textbf{RQ2}). Specifically, the adversarial advantage is gained through access to increasingly more quantum information under $\mathsf{M}$ and $\mathsf{Q}$ (\textbf{RQ2}), and can be exploited by quantum algorithms (\textbf{RQ1}), while classical outcome-only access under $\mathsf{C}$ eliminates this quantum advantage.
    \item We establish a bound for the gap between the theoretical optimal membership inference advantage and empirical adversary gain with limited resources and finite surrogate models; our empirical experiments confirm this relationship. We show that the probabilistic nature of QML models contributes to this gap, but does not prevent MIAs; it makes it harder for adversaries to exploit quantum information for privacy inference empirically, even when a theoretical advantage exists (\textbf{RQ2}).
\end{itemize}

\section{Related Works and Preliminaries}\label{sec:lit}

\subsection{Quantum Machine Learning}\label{ssec:lit-qml}

In general, Quantum Machine Learning (QML) \cite{Biamonte2017QML,Cerezo2022Challenges} aims to utilize quantum computation for data representation, model evaluation, or optimization in machine learning (ML) tasks \cite{Nowmi2026SoK}. Since its first proposal, QML has developed into a wide spectrum of algorithms, including quantum kernel methods \cite{Schuld2019QML,Havlivcek2019Supervised}, variational quantum eigensolvers (VQE) \cite{Kandala2017VQE}, variational quantum circuits (VQCs) \cite{Mitarai2018Circuit,Benedetti2019PQC}, quantum neural networks (QNNs) \cite{Cerezo2021VQA} and hybrid quantum-classical neural networks (HQNNs) \cite{Mari2020HQNN,Liu2021HQNN}, etc.\footnote{Many notations, e.g., parameterized quantum circuits (PQCs) and variational quantum circuits (VQCs) are used interchangeably by existing works.} In this paper, we focus exclusively on QNNs for supervised learning tasks, which constitute the principal approaches to machine learning on contemporary quantum processors \cite{Mitarai2018Circuit,Benedetti2019PQC,Cerezo2021VQA}.

\paragraph{Preliminaries to quantum computation} We aim here to introduce a minimal notation necessary for the paper rather than a more detailed overview of QML. The fundamental unit of quantum computation is a \emph{qubit}:
\begin{equation}
    \ket{\psi} = \alpha\ket{0} + \beta\ket{1}
\end{equation}
in which $\alpha,\beta \in \mathbb{C}$, $|\alpha|^2 + |\beta|^2 = 1$. Measuring $\ket{\psi}$ returns $\ket{0}$ with probability $|\alpha|^2$ and $\ket{1}$ with probability $|\beta|^2$ \cite{NC2010}. More generally, an $n$-qubits pure state is a unit vector in a $2^n$-dimensional Hilbert space, written as:
\begin{equation}
    \ket{\psi} = \sum_{i \in \{0,1\}^n} \alpha_i \ket{i}, \quad \sum_i |\alpha_i|^2 = 1.
\end{equation}
Equivalently, quantum states can be represented as \emph{density operators}. A pure state $\ket{\psi}$ corresponds to $\rho = \ketbra{\psi}{\psi}$, which allows us to represent mixed states and quantum noises.

A quantum gate for an $n$-qubit system is represented by a unitary matrix $U \in \mathbb{U}(2^n)$, such that:
\begin{equation}
    U^\dagger U = U U^\dagger = I,
\end{equation}
which transforms a state as:
\begin{equation}
    \ket{\psi} \mapsto U\ket{\psi} \quad \text{or} \quad
    \rho \mapsto U \rho U^\dagger.
\end{equation}
Common single-qubit gates include the Pauli gates $X$, $Y$, $Z$, the Hadamard gate $H$, and parameterized rotation gates $R_X(\theta)$, $R_Y(\theta)$, and $R_Z(\theta)$, while multi-qubit gates generate \emph{entanglement} between qubits that allow information interaction \cite{NC2010}. A quantum circuit is a sequence of such gates.

\paragraph{Variational quantum circuits} A dominating paradigm of QML designed for noisy intermediate-scale quantum computing (NISQ) systems are \emph{variational quantum circuits} (VQCs), which utilize quantum gates with \emph{trainable} parameters (e.g., single-qubit rotations) that allow circuits to learn an input-output pair with objective functions, similar to classical ML models \cite{Mitarai2018Circuit,Benedetti2019PQC,Cerezo2021VQA}. A VQC is represented as a parameterized unitary $U({\bm \theta})$, ${\bm \theta} = (\theta_1,\dots,\theta_p)$. Given an input state $\rho_{\mathrm{in}}$, the circuit produces:
\begin{equation}
    \rho_{\bm \theta} = U({\bm \theta}) \rho_{\mathrm{in}} U({\bm \theta})^\dagger.
\end{equation}
To receive the output, VQCs \emph{measure} $\rho_{\bm\theta}(x)$ through an operator $M$, which collapses the superposition into classical outcomes, e.g., measuring through the Pauli-$Z$ gate produces a real-valued expectation $M(\rho_{\bm\theta}(x)) \in [-1,1]$ \cite{NC2010}. In summary, the VQC is the pipeline:
\begin{equation}
x \xrightarrow{U_{\mathrm{enc}}(x)}
\ket{\phi(x)} \xrightarrow{U({\bm \theta})}
\rho_{\bm \theta}(x) \xrightarrow{M}
M(\rho_{\bm\theta}(x)).
\end{equation}
A common practice to reduce the uncertainty on the model output is through preparing and measuring the same circuit for multiple shots, which we denote as $M(\rho_{\bm\theta}(x))^{\otimes k}$. Each circuit execution is conventionally called a \emph{shot}, and a finite number of shots therefore yields a stochastic estimate of the underlying expectation value.

VQCs are typically trained in an optimization loop alternating between quantum circuit evaluations and classical parameter updates \cite{Mitarai2018Circuit,Cerezo2021VQA}. Consider a supervised training dataset $D = {(x_i,y_i)}_{i=1}^{N}$. The empirical training objective is:
\begin{equation}
    {\bm\theta} = \argmin_{\bm\theta} \frac{1}{N} \sum_{i=1}^{N} \ell \big( M(\rho_{\bm\theta}(x_i)), y_i \big),
\end{equation}
the circuit parameters can therefore be optimized following gradient descent:
\begin{equation}
    {\bm \theta}^{(t+1)} = {\bm \theta}^{(t)} -
    \eta \nabla_{\bm \theta} \ell_D \left( {\bm \theta}^{(t)} \right),
\end{equation}
or with common optimizers, e.g., stochastic gradient descent (SGD) or Adam. The derivatives of expectation values of rotation gates can be calculated under the parameter-shift rule \cite{Biamonte2017QML,Schuld2019QML,NC2010}.

\paragraph{Encoding classical data} To utilize QC for QML tasks, any classical input data (e.g., image and text) need to be converted into a quantum state. An encoding circuit $U_{\mathrm{enc}}(x)$ prepares:
\begin{equation}
    \ket{\phi(x)} = U_{\mathrm{enc}}(x) \ket{0}.
\end{equation}
Two common approaches are \emph{angle encoding} and \emph{amplitude encoding} \cite{Benedetti2019PQC,Nowmi2026SoK}. Angle encoding represents classical features through rotation angles, i.e.:
\begin{equation}
    \ket{\phi(x)} = \bigotimes_{i=1}^{n} R_Y(x_i) \ket{0}.
\end{equation}
Amplitude encoding instead represents (normalized) feature vectors through the amplitudes of a quantum state:
\begin{equation}
    \ket{\phi(x)} = \sum_{i=0}^{d-1} \widetilde{x}_i \ket{i}, \quad
    \sum{i=0}^{d-1} |\widetilde{x}_i|^2 = 1.
\end{equation}
A $d$-dimensional feature be encoded with $\lceil \log_2 d \rceil$ qubits. Our result does not depend on specific encoding schemes.

\paragraph{Quantum Neural Networks} We now give a clear definition to the quantum neural network (QNN) models studied in this paper as follows:

\begin{definition}[Quantum Neural Network \cite{Benedetti2019PQC,Cerezo2021VQA}]\label{def:qml}
    A \emph{quantum neural network} (QNN) is a supervised classifier in which a classical input $x$ is encoded into an $n$-qubit state and subsequently processed by a trainable VQC. Given trainable parameters ${\bm \theta}$, its pre-measurement state is:
    \begin{equation}
        \rho_{\bm \theta}(x)
        = U({\bm \theta}) U_{\mathrm{enc}}(x) \rho_0 U_{\mathrm{enc}}(x)^\dagger U({\bm \theta})^\dagger,
    \end{equation}
    in which $\rho_0 = \ketbra{0^n}{0^n}$ for an $n$-qubit circuit. A \emph{positive operator-valued measure} (POVM) $M = \{E_y\}_y$ maps this state to classical outcomes $M(\rho_{\bm\theta}(x))$, which satisfies:
    \begin{equation}
        p(y \mid \rho_{\bm\theta}) = \trace (E_y \rho_{\bm\theta}),
    \end{equation}
    that is, the probability of each possible measurement outcome is the projected trace.
\end{definition}

\subsection{Membership Inference Attacks}\label{ssec:lit-mia}

Parallel to QML, ML models have long been understood to be susceptible to privacy inference attacks that allow adversaries to infer privacy-sensitive statistics of the model and its training data. Most significantly, \emph{membership inference attacks} (MIAs) aim to infer whether a data sample is present in the training set of a model or not\cite{Shokri2017MIA,Yeom2018Overfit,Carlini2022MIA,Ye2022Enhanced}.

\paragraph{Population-level MIAs} MIA against ML is first proposed by Shokri et al. in \cite{Shokri2017MIA}, in which the adversary has access to surrogate data sampled from the same data distribution $\mathcal{Z}$ as the model, and trains multiple surrogate (shadow) models. An attack model is then trained to classify the membership status of the target sample given its behavior on the surrogate models. Following their line of research, Yeom et al. \cite{Yeom2018Overfit} use the model confidence as the feature for the attack model, as greater confidence indicates a higher probability of membership due to the model overfitting on the training data.

\paragraph{Instance-level MIAs} Beyond population-level MIAs that utilize a singular attack model to distinguish membership status on all target samples, an adversary can design instance-level MIA where hypothesis testing is performed for each individual target sample independently \cite{Carlini2022MIA,Ye2022Enhanced}.
Carlini et al. in \cite{Carlini2022MIA} proposes the likelihood-ratio attack (LiRA), which trains several shadow models divided as two worlds: where the target $z=(x,y)$ is or is not a member of the surrogate model's training set; the shadow models' output probabilities on the target sample at label $y$ are then fitted into two probability distributions, conditioned on the membership. An adversary determines the membership status by deciding which distribution is more likely to result in the target model's behavior. Ye et al. \cite{Ye2022Enhanced} follow a similar intuition, but use the model confidence on $(x,y)$.

Following Carlini et al. \cite{Carlini2022MIA}, we define membership inference attack against ML/QML in Game \ref{game:mia} as a game between the challenger (model owner) $\Chal$ and the adversary $\Adv$: $\Chal$ provides two models trained with the same learning algorithm on neighboring datasets that differ only on the target sample $z$, $\Adv$ in turn develops a decision rule $h: \mathcal{Z} \times {\bm\Theta} \to \{0,1\}$ that aims to distinguish between the two models.

\begin{game}[ht]
\caption[F]{Membership Inference Game ($\mathsf{MI}$)}\label{game:mia}
    \KwIn{Target sample $z=(x,y)$, training algorithm $\mathcal{A}$, neighboring datasets $D_0, D_1 \in \mathcal{Z}^n$.}
    \uIf{$b = 0$}{
        ${\bm\theta}_b \gets \mathcal{A}(D_0 \mid z \notin D_0)$
        \tcp*[l]{non-member}
    }
    \Else{
        ${\bm\theta}_b \gets \mathcal{A}(D_1 \mid z \in D_1)$
        \tcp*[l]{member}
    }
    $\hat{b} \gets h(z, {\bm\theta})$ \;
    \KwOut{$\mathbf{1}[\hat{b} = b]$, membership inference advantage $L := 2\Pr[\hat{b}=b]-1$, $L \in [0,1]$.}
\end{game}

\subsection{Privacy Inference Attacks Against QML}\label{ssec:lit-priv-qml}

As discussed in Sec. \ref{sec:intro}, the privacy risks of ML are often inherited in QML models \cite{Heredge2025Character,Su2026MIAQMU}. Early work largely approached QML privacy from the defensive perspective. Zhou et al. \cite{Zhou2017DPQC} proposed the notion of \emph{quantum differential privacy} (QDP), which establishes DP for QC tasks. Watkins et al. \cite{Watkins2023QDP} further applied QDP to training hybrid quantum-classical models to bound the membership leakage of the model's training data.

Heredge et al. \cite{Heredge2025Character} provide a foundational characterization of privacy risks in QML, which demonstrates a privacy-trainability trade-off for QML models. The authors establish the connection between the model's susceptibility to both quantum and classical data reconstruction attacks with VQC gradients to the dynamical Lie algebra (DLA) of the VQC.
Kumar et al. \cite{Kumar2023Expressive} give a related result, arguing that expressive quantum encodings make data reconstruction difficult and suggesting that QML provides inherent resilience to data reconstruction attacks. Importantly, however, the adversarial information exposed by the model remains classical (e.g., gradients and architecture circuit descriptions), even when QC are used in the attack.

The closest work to our paper is Su et al. \cite{Su2026MIAQMU}, who proposed a (population-level) membership inference attack against QNNs and HQNN models, which can subsequently be mitigated through (quantum) machine unlearning (MU) \cite{Bourtoule2021SISA,Nguyen2025Survey,Shaik2025QMU,Crivoi2026QMU}.
The authors experimentally demonstrate membership leakage in both noiseless simulations and on quantum hardware. Similar to Heredge et al. \cite{Heredge2025Character}, their proposed attack also takes measured QNN outcomes as inputs, and additionally demonstrates that the number of measurement shots affects the observable membership leakage.

Overall, while existing works on MIAs against QML models \cite{Heredge2025Character,Su2026MIAQMU} establish the existence of privacy leakage through the measured \emph{classical} outcome, they do not consider an adversary with access to \emph{quantum} states as we discussed in Sec. \ref{sec:intro}. A quantum-capable adversary with quantum-native access may apply alternative measurements and potentially extract information discarded by a fixed readout \cite{Zurek2003Deco}. Our work provides the first explicit examination of these stronger access regimes and their implications for privacy leakage through membership inference attacks in QML.

\section{Quantum-Native Access and Membership Risk}\label{sec:theory}

\begin{table*}[t]
    \centering
    \small
    \renewcommand{\arraystretch}{1.1}
    \caption{Access regimes considered in Sec. \ref{sec:theory} and \ref{sec:eval}.}\label{tab:access}
    \begin{tabular}{lll}
        \toprule
        \textbf{Access Regime} & \textbf{Interface} & \textbf{Notes} \\
        \midrule
        Classical outcome-only $\mathsf{C}_k$ &
        \parbox[t]{170pt}{
        The user queries the QNN with classical input $x$. The model owner prepares $k$ copies of $x$, the QNN outputs the pre-measurement state $\rho_{\bm\theta}(x)^{\otimes k}$.
        The model owner fixes measurement operator $M_0 = \{E_y\}_{y \in \mathcal{Y}}$ and returns only the classical outcomes $\hat{y}^k \sim M(\rho_{\bm\theta}(x))^{\otimes k}$
        } &
        \parbox[t]{170pt}{The adversary can perform arbitrary post-processing on the classical outcomes $\hat{y}^k$ but cannot access the quantum state $\rho_{\bm \theta}(x)$. Most existing works with ``classical'' adversaries follow this access regime \cite{Heredge2025Character,Su2026MIAQMU}.} \\
        \\
        
        Programmable measurement $\mathsf{M}_k$ &
        \parbox[t]{170pt}{
        The user submits $x$ along with a measurement $M_u$ (we restrict the user-selected $M_u$ to be repeated Pauli measurements) to the model owner for $x$, the model owner returns $M_u(\rho_{\bm\theta}(x))^{\otimes k}$.
        } &
        \parbox[t]{170pt}{
        While the adversary still only receives classical outcomes, they can decide the measurements used to potentially achieve higher membership inference advantage. Emerging quantum cloud services can provide the programmable measurement for $\mathsf{M}_k$ \cite{Devitt2016Cloud,Fu2024QuantumLeak}.
        } \\
        \\

        Quantum state access $\mathsf{Q}_k$ &
        \parbox[t]{170pt}{
        The user submits $x$, for which the model owner returns $\rho_{\bm\theta}(x)^{\otimes k}$.
        } &
        \parbox[t]{170pt}{
        The adversary receives the pre-measurement quantum state without information loss from POVM $M$, as discussed by Zurek \cite{Zurek2003Deco} and Su et al. \cite{Su2026MIAQMU}. Therefore, the adversary has strictly more information than regime $\mathsf{M}_k$ and by extension $\mathsf{C}_k$.
        } \\
        \\


        
        \bottomrule
    \end{tabular}
\end{table*}

\subsection{Threat Model}\label{ssec:threat-model}

\paragraph{Adversary's goal} For a QNN given by Def. \ref{def:qml}, we consider an adversary $\Adv$ that performs membership inference attacks (Game \ref{game:mia}) against the model aims to infer whether a target sample $z = (x,y)$ is in the training set of the model $\bm\theta$, i.e., $\mathbf{1}_D(z)$.

\paragraph{Adversary's abilities} We assume an adversary has quantum computing abilities and can train, execute and measure quantum circuits and VQCs. We follow the general assumption that the adversary can sample surrogate dataset(s) from the same distribution as the target model \cite{Shokri2017MIA,Carlini2022MIA,Su2026MIAQMU}.

\paragraph{Access regimes} We discuss three \emph{access regimes} of the QNN, summarized in Table \ref{tab:access}. We aim to demonstrate that, with increasingly more access to quantum information, the adversary can achieve more membership inference advantage (\textbf{RQ2}), i.e., $\mathsf{C} \preceq \mathsf{M} \preceq \mathsf{Q}$.


\subsection{Membership Inference Advantages under Different Access Regimes}\label{ssec:mia-theory}

Before discussing MIA privacy leakages with quantum information, we focus on the commonly employed classical access regime $\mathsf{C}$ \cite{Heredge2025Character,Kumar2023Expressive,Su2026MIAQMU}. Intuitively, for (\textbf{RQ1}) Immediately, we arrive at the first result formalized in Theorem \ref{th:post-measure-equiv}.

\begin{theorem}[Post-measurement equivalence]\label{th:post-measure-equiv}
    Under the access regime $\mathsf{C}_k$ in which the adversary receives only classical outputs, for any quantum-capable adversary with quantum computing abilities, there exists an adversary with only classical computing abilities that has the same membership inference advantage over random guessing.
\end{theorem}
\begin{proof}
    From Def. \ref{def:state-discrepancy}, given the POVM operator $M_0 = \{E_y\}_{y \in \mathcal{Y}}$, the user receives the output $y^k$ following the distribution:
    \begin{equation}
        y^k \sim p_b^{(k)}(y_1,\ldots,y_k) =
        \mathbb E \left[ \prod_{j=1}^k \trace (E_{y_j}\rho_{{\bm\theta}_b}(x)) \right].
    \end{equation}
    Membership inference from the measurement outputs is therefore a binary hypothesis testing between $p_0^{(k)}$ and $p_1^{(k)}$. Under equal priors, the optimal success probability is given by:
    \begin{equation}
        P_{\mathrm{succ}}^* = \frac{ 1 + \TV \left( p_0^{(k)},p_1^{(k)} \right) }{2},
    \end{equation}
    and the corresponding membership inference advantage over random guessing:
    \begin{equation}
        L_\mathsf{C}^{(k)} = \TV \left( p_0^{(k)},p_1^{(k)} \right)
    \end{equation}
    The same conclusion follows directly from quantum hypothesis testing. Write the classical outputs as \emph{diagonal states}:
    \begin{equation}
        \sigma_b^{(k)} = \sum_{y^k} p_b^{(k)}(y^k) \ketbra{y^k}{y^k},
    \end{equation}
    with trace distance:
    \begin{equation}
    \begin{aligned}
        \frac{1}{2} \left\| \sigma_1^{(k)}-\sigma_0^{(k)} \right\|_1
        &= \frac{1}{2} \sum_{y^k} \left| p_1^{(k)}(y^k)-p_0^{(k)}(y^k) \right| \\
        &= \TV \left( p_0^{(k)}, p_1^{(k)} \right). \\
    \end{aligned}
    \end{equation}
    By Helstrom's theorem \cite{Helstrom1969Detection,Bae2015Discri}, the optimal success probability for any unrestricted quantum measurement to distinguish these states is:
    \begin{equation}
    \begin{aligned}
        P_{\mathrm{succ}, \mathrm{QC}}^*
        &= \frac{1}{2} + \frac{1}{4} \left\| \sigma_1^{(k)}-\sigma_0^{(k)} \right\|_1 \\
        &= \frac{ 1 + \TV \left( p_0^{(k)},p_1^{(k)} \right) }{2}.
    \end{aligned}
    \end{equation}
    Therefore, under $\mathsf{C}_k$, quantum and classical adversaries have identical optimal membership advantage with only classical measurement outcomes.
\end{proof}

\begin{tbox}
\begin{observation}[Quantum advantage under $\mathsf{C}_k$]\label{obs:advantage-C}
    Theorem \ref{th:post-measure-equiv} clearly establishes that a quantum-capable adversary does not possess quantum advantage for $\mathsf{MI}$ under the classical-outcome-only regime $\mathsf{C}_k$. Quantum algorithms may nevertheless reduce the computational complexity of classical statistical testing tasks with additional oracle access to the underlying distributions \cite{Bravyi2011Distributions,Gilyen2019Distributional}.
\end{observation}
\end{tbox}

Recall the LiRA membership inference attack proposed by Carlini et al. \cite{Carlini2022MIA}, in which the adversary aims to distinguish (non-quantum) model outputs of two worlds: those trained with and without the target sample as a member of the training set. The same intuition can be adopted for MIAs against QML models; we define:

\begin{definition}[Membership-Conditioned State Discrepancy]\label{def:state-discrepancy}
    For a QNN model given by Def. \ref{def:qml}, denote the membership-conditioned expected pre-measurement state as:
    \begin{equation}\label{eq:omega}
        \Omega_b(x) = \mathbb{E}[\rho_{{\bm\theta}_b}(x)^{\otimes k}],
    \end{equation}
    in which ${\bm\theta}_b$ is the QML model trained under bit $b$ in Game \ref{game:mia}, shared through $k$ copies of $x$. $\Omega_b(x)$ takes expectation over different ${\bm\theta}_b$ trained with different surrogate datasets. The membership-conditioned (pre-measurement) state discrepancy induced by $x$ is therefore:
    \begin{equation}\label{eq:delta}
        \Delta(x) = \Omega_1(x) - \Omega_0(x),
    \end{equation}
    that is, how much $x$ changes the model's output state.
\end{definition}

\begin{theorem}[Access regime hierarchy]\label{th:hierachy}
    For any target sample $z=(x,y)$ whose membership status we are interested in, fix the target model $\bm\theta$, we have:
    \begin{equation}
        L_\mathsf{C}(z; \bm\theta) \leq L_\mathsf{M}(z; \bm\theta) \leq L_\mathsf{Q}(z; \bm\theta),
    \end{equation}
    that is, the access regimes have a strict hierarchy in quantum advantage.
\end{theorem}
\begin{proof}
    By Theorem \ref{th:post-measure-equiv}, we have:
    \begin{equation}\label{eq:classical-opt}
        L_\mathsf{C}(z; \bm\theta) = \TV \left( p_0^{(k)},p_1^{(k)} \right)
    \end{equation}
    Similarly, we can write:
    \begin{equation}\label{eq:measurement-opt}
        L_\mathsf{M}(z; \bm\theta) =
        \sup_{M} \TV \left( M\big(\Omega_0(x)\big), M\big(\Omega_1(x)\big) \right)
    \end{equation}
    Further, under $\mathsf{Q}_k$, the adversary receives the $k$ quantum states. By Helstrom's theorem \cite{Helstrom1969Detection}, now applied to a mixed state:
    \begin{equation}\label{eq:quantum-opt}
        L_\mathsf{Q}(z; \bm\theta) = \frac{1}{2} \left\| \Delta(x) \right\|_1,
    \end{equation}
    $\Delta(x)$ is given by Eq. \ref{eq:delta}.
    The adversary obtains strictly more information from $\mathsf{C}_k$ to $\mathsf{M}_k$ to $\mathsf{Q}_k$ as the selection of measurement operators expands. The inequalities naturally follow.
\end{proof}

\begin{tbox}
\begin{observation}[Quantum advantage through multi-shots]\label{obs:advantage-k}
    Beyond Theorem \ref{th:hierachy}, we also note that the adversary gains are non-decreasing in the model measurement budget $k$ for the same access regime $\mathsf{I}$, formally:
    \begin{equation}
        L_\mathsf{I}^{(k+1)} \geq L_\mathsf{I}^{(k)},
    \end{equation}
    that is, outcomes from more shots provide strictly more information, which allows $\Adv$ to design stronger MIAs.
\end{observation}
\end{tbox}

\subsection{Membership Inference Advantages in Practice}\label{ssec:mia-empirical}

In Theorem \ref{th:hierachy} (Eq. \ref{eq:classical-opt} through \ref{eq:quantum-opt}), we provide the closed-form optimal membership inference advantage for all access regimes we considered. However, this optimal MI advantage is information-theoretic and is achieved only when the adversary knows the entire membership-conditioned outcome distributions. We now study the question more closely and ask:
\begin{quote}
    Given an adversarial resource constraint, when the adversary has imperfect knowledge of the membership-conditioned outcome distributions, how much membership inference advantage remains?
\end{quote}
We formalize it as follows:

\begin{theorem}[Finite-Resource Attacks]\label{thm:finite}
    Fix target sample $z=(x,y)$, access regime $I\in\{ \mathsf{M}, \mathsf{Q} \}$ and the $k$-shot measurement budget $k$. Assume $\Adv$ can only train $K$ pairs of shadow/surrogate models $({\bm\theta}^{s}_{0, 1}, {\bm\theta}^{s}_{1, 1})$ through $({\bm\theta}^{s}_{0, K}, {\bm\theta}^{s}_{1, K})$.
    Denote $\widehat{L}_\mathsf{I}(z, {\bm\theta} \mid K)$ as the \emph{empirical} membership advantage under finite resource constraints. The \emph{gap} between theoretical and empirical membership inference advantage is:
    \begin{equation}
        \varepsilon_\mathsf{I} (z, {\bm\theta} \mid K) :=
        L_\mathsf{I}(z, {\bm\theta}) - \widehat{L}_\mathsf{I}(z, {\bm\theta} \mid K).
    \end{equation}
    Then we have:
    \begin{equation}
        \varepsilon_\mathsf{I} (z, {\bm\theta} \mid K) \leq
        \left\| \widehat{\Delta}(x) - \Delta(x) \right\|_1 ,
    \end{equation}
    in which:
    \begin{equation}
        \widehat{\Delta}(x) = \frac{1}{K} \sum_{i=1}^{K}
        \left[ \rho_{{\bm\theta}^{s}_{1, i}}(x) - \rho_{{\bm\theta}^{s}_{0, i}}(x) \right].
    \end{equation}
    The advantage gap between the theoretical optimum and the empirical result is therefore bounded by how well the $K$ surrogate pairs represent the true population of trained models.
\end{theorem}

\begin{proof}
    Similar to Def. \ref{def:state-discrepancy}, we write:
    \begin{equation}\label{eq:omega-hat}
        \widehat{\Omega}_{b}(x) = \frac{1}{K} \sum_{i=1}^{K} \rho_{{\bm\theta}_{b, i}}(x)^{\otimes k},
    \end{equation}
    then:
    \begin{equation}
        \widehat{\Delta}(x) = \widehat{\Omega}_1(x) - \widehat{\Omega}_0(x)
    \end{equation}
    is the \emph{observed} membership-conditioned discrepancy. The membership advantage of an attack based on observations from the $K$ surrogate models is:
    \begin{equation}
        \widehat{L}_\mathsf{I}(z, {\bm\theta} \mid K) =
        \TV \big( 
        M_u ( \widehat{\Omega}_1 ),
        M_u ( \widehat{\Omega}_0 )
        \big).
    \end{equation}
    Recall Theorem \ref{th:post-measure-equiv}, trace distance is contractive, which gives:
    \begin{equation}
    \begin{aligned}
        &\abs{
        \TV \big( M_u ( \widehat{\Omega}_1 ), M_u ( \widehat{\Omega}_0 ) \big) -
        \TV \big( M_u ( {\Omega}_1 ), M_u ( {\Omega}_0 ) \big)
        } \\
        &\leq \frac{1}{2} \left\| \Delta - \widehat{\Delta} \right\|_1 ,
    \end{aligned}
    \end{equation}
    and similarly:
    \begin{equation}
    \begin{aligned}
        &\abs{
        \sup_M \TV \big( M ( {\Omega}_1 ), M ( {\Omega}_0 ) \big) -
        \TV \big( M_u ( {\Omega}_1 ), M_u ( {\Omega}_0 ) \big)
        } \\
        &\leq \frac{1}{2} \left\| \Delta - \widehat{\Delta} \right\|_1 ,
    \end{aligned}
    \end{equation}
    in which the first term is equal to $L_\mathsf{I}(z; \bm\theta)$ for both $\mathsf{M}$ and $\mathsf{Q}$. By the triangle inequality:
    \begin{equation}
        \abs{L_\mathsf{I}(z, {\bm\theta}) - \widehat{L}_\mathsf{I}(z, {\bm\theta} \mid K)}
        \leq \left\| \Delta(x) - \widehat{\Delta}(x) \right\|_1 ,
    \end{equation}
    which proves the theorem.
\end{proof}

\begin{tbox}
\begin{observation}[Lower bound on $\varepsilon_\mathsf{I}$]\label{obs:lower-bound-eps}
    Note that Theorem \ref{thm:finite} deliberately avoids giving a lower bound on $\varepsilon_\mathsf{I} (z, {\bm\theta} \mid K)$, as we can easily construct a one-shot observation ($k=1$) result such that $L_\mathsf{I} < 1$ and $\widehat{L}_\mathsf{I} = 1$. $\varepsilon_\mathsf{I} (z, {\bm\theta} \mid K)$ should therefore be studied in the form of \emph{expectations}, which we expand in Corollary \ref{cor:exp-bound}.
\end{observation}
\end{tbox}

Theorem \ref{thm:finite} provides an upper bound for the membership inference advantage obtained by $\Adv$ compared to the optimal theoretical advantage given by Theorem \ref{th:hierachy}. We further derive the following corollary:

\begin{corollary}[Expectation Bound on the Finite-Resource Attacks Gap]\label{cor:exp-bound}
    Under Theorem \ref{thm:finite}, let $d$ be the dimension of $\rho_{\bm\theta}(x)$, $\rho_{\bm\theta}(x)^{\otimes k}$ has dimension $d^k$. We have:
    \begin{equation}\label{eq:exp-gap}
        \mathbb{E} \left[ \varepsilon_\mathsf{I} (z, {\bm\theta} \mid K) \right] \leq \sqrt{\frac{2d^k}{K}}.
    \end{equation}
\end{corollary}
\begin{proof}
    Recall:
    \begin{equation}
    \begin{aligned}
        \widehat{\Delta}(x) &= \widehat{\Omega}_1(x) - \widehat{\Omega}_0(x) \\
        &= \frac{1}{K} \sum_{i=1}^{K}
        \underbrace{\left[
        \rho_{{\bm\theta}_{1, i}}(x)^{\otimes k} - \rho_{{\bm\theta}_{0, i}}(x)^{\otimes k}
        \right]}_{\widehat{\Omega}_{b, i}(x)}.
    \end{aligned}        
    \end{equation}
    Then $\mathbb{E}[\widehat{\Omega}_{b, i}(x)] = \Delta(x)$, each term is independent across surrogate pairs. Recall that the Hilbert–Schmidt norm between two density operators is at most $\sqrt{2}$, then:
    \begin{equation}
        \mathbb{E} \left[ \| \widehat{\Delta} - \Delta \|_2 \right] \leq \sqrt{2/K}.
    \end{equation}
    Because the received $k$-shot state has dimension $d^k$, $\|A\|_1\leq\sqrt{d^k}\|A\|_2$. By Jensen's inequality:
    \begin{equation}
        \mathbb{E} \left[ \varepsilon_\mathsf{I}(z, {\bm\theta} \mid K) \right] \leq
        \mathbb{E} \left[ \| \widehat{\Delta} - \Delta \|_1 \right] \leq
        \sqrt{\frac{2d^k}{K}},
    \end{equation}
    which proves the corollary.
    
\end{proof}

\begin{tbox}
\begin{observation}[Interpreting the membership inference advantage gap]\label{obs:gap}
    Together, Theorem \ref{thm:finite} and Corollary \ref{cor:exp-bound} demonstrate that there exists an (inescapable) gap between the empirical membership inference advantage and the theoretical optimum shown in Theorem \ref{th:hierachy}. This gap is controlled by two hyperparameters: 
    \begin{itemize}
        \item A larger $K$ (more surrogate models) allows $\Adv$ to close the gap, as more surrogate models allow the adversary to better approximate the overall membership-conditioned state distribution;
        
        \item Counterintuitively, a larger $k$ (more observation shots on target/surrogate models) increases the \emph{expectation} of the gap, yet as we established in Observation \ref{obs:advantage-k}, more shots increases the adversarial membership inference advantage under any access regime $\mathsf{I}$. 
        
    \end{itemize}
    The inequality in Eq. \ref{eq:exp-gap} should therefore be understood as a statement on the \emph{learnability} of the theoretical membership advantage, and not a statement on the \emph{performance} of the empirical attack. We demonstrate this tension empirically in Sec. \ref{sec:eval}.
\end{observation}
\end{tbox}

\section{Evaluation}\label{sec:eval}

\subsection{Experimental Setup}\label{ssec:setup}

\paragraph{Dataset and data pre-processing} We use two subsets of the MNIST dataset \cite{LeCun1998MNIST} for binary classification tasks of distinguishing (\rom{1}) ``0'' vs ``1'' and (\rom{2}) ``3'' vs ``8'', following Su et al. \cite{Su2026MIAQMU}. 
We normalize the images in MNIST dataset to $28 \times 28$ pixels, each with values in the range $[-1,1]$, perform principal component analysis (PCA), and select the top-4 component elements as input features.

\paragraph{Models} We plot the circuit of the target and surrogate model QNNs in Fig. \ref{fig:circuit}. The input component $x_i$ is encoded by $R_Y(\pi x_i)$.  Three trainable layers apply. The model has 4 qubits with 3 variational layers that apply $R_Y$ and $R_Z$ on every qubit, followed by a ring of \texttt{CNOT} gates, 24 trainable angles in total, comparable to the models used in Heredge et al. \cite{Heredge2025Character}. The approved prediction is the expectation of $Z_0$, represented as $ZIII$. In our experiments, we study $k \in \{1, 2, 4, 8, 16, 32\}$-shot measurement outcomes.

All QNN and adversary hypothesis testing are simulated using the PennyLane library \cite{Bergholm2018PennyLane}. All models are trained for 80 epochs with Adam optimizer \cite{Kingma2015Adam} for a $5 \times 10^{-2}$ learning rate. We use binary cross-entropy as the loss function. We also use an initialization scale of $0.2$ for all models.

\begin{figure*}[t]
    \centering
    \includegraphics[width=0.9\linewidth]{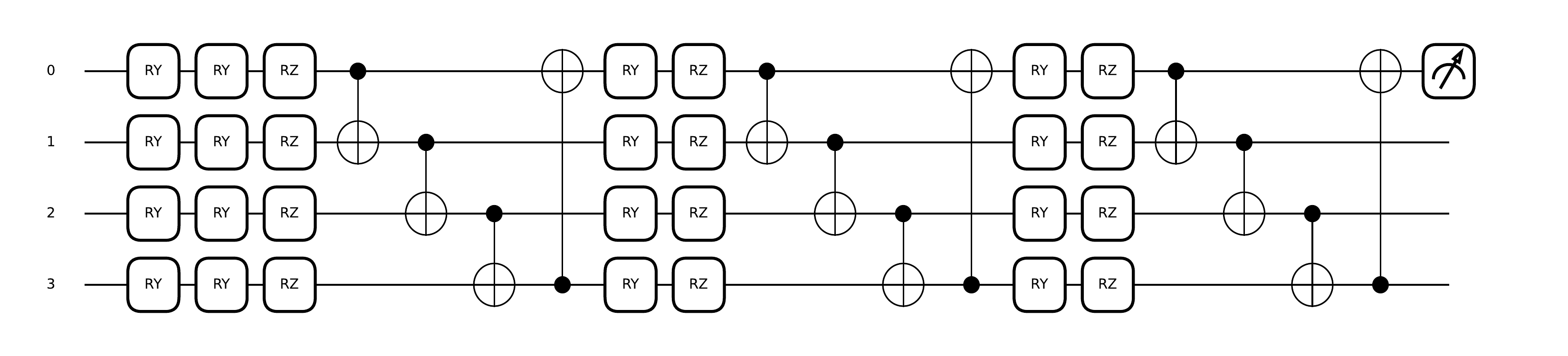}
    \caption{The target/surrogate QNN model circuit used in Sec. \ref{sec:eval}.}
    \label{fig:circuit}
\end{figure*}

\begin{figure}[t]
    \centering
    \captionsetup[subfigure]{justification=centering}
    \begin{subfigure}[b]{.49\linewidth}
        \centering
        \includegraphics[width=\textwidth]{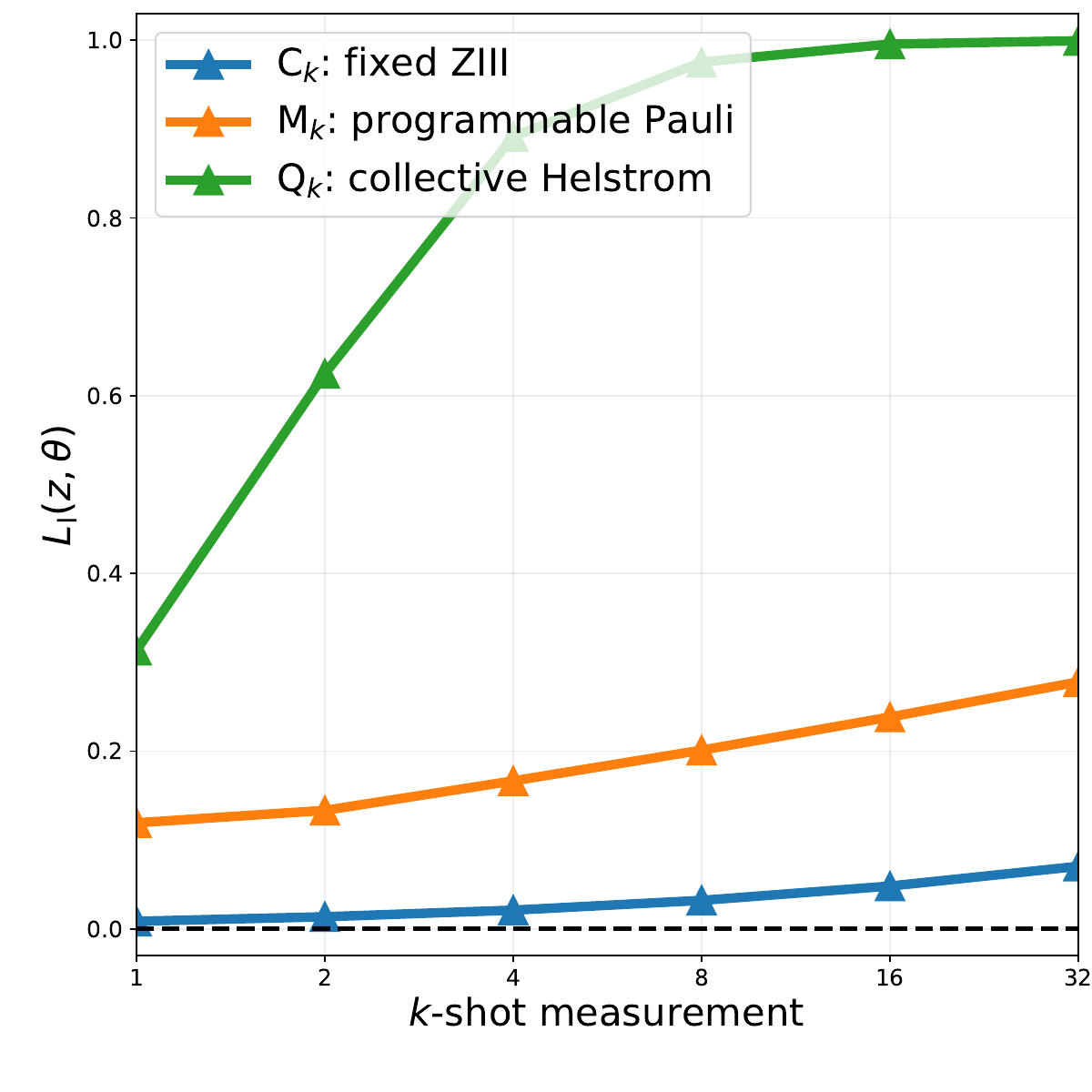}
        \caption{``0'' vs. ``1'', $K=64$.}
        \label{sfig:0v1_fix_surr}
    \end{subfigure}
    \begin{subfigure}[b]{.49\linewidth}
        \centering
        \includegraphics[width=\textwidth]{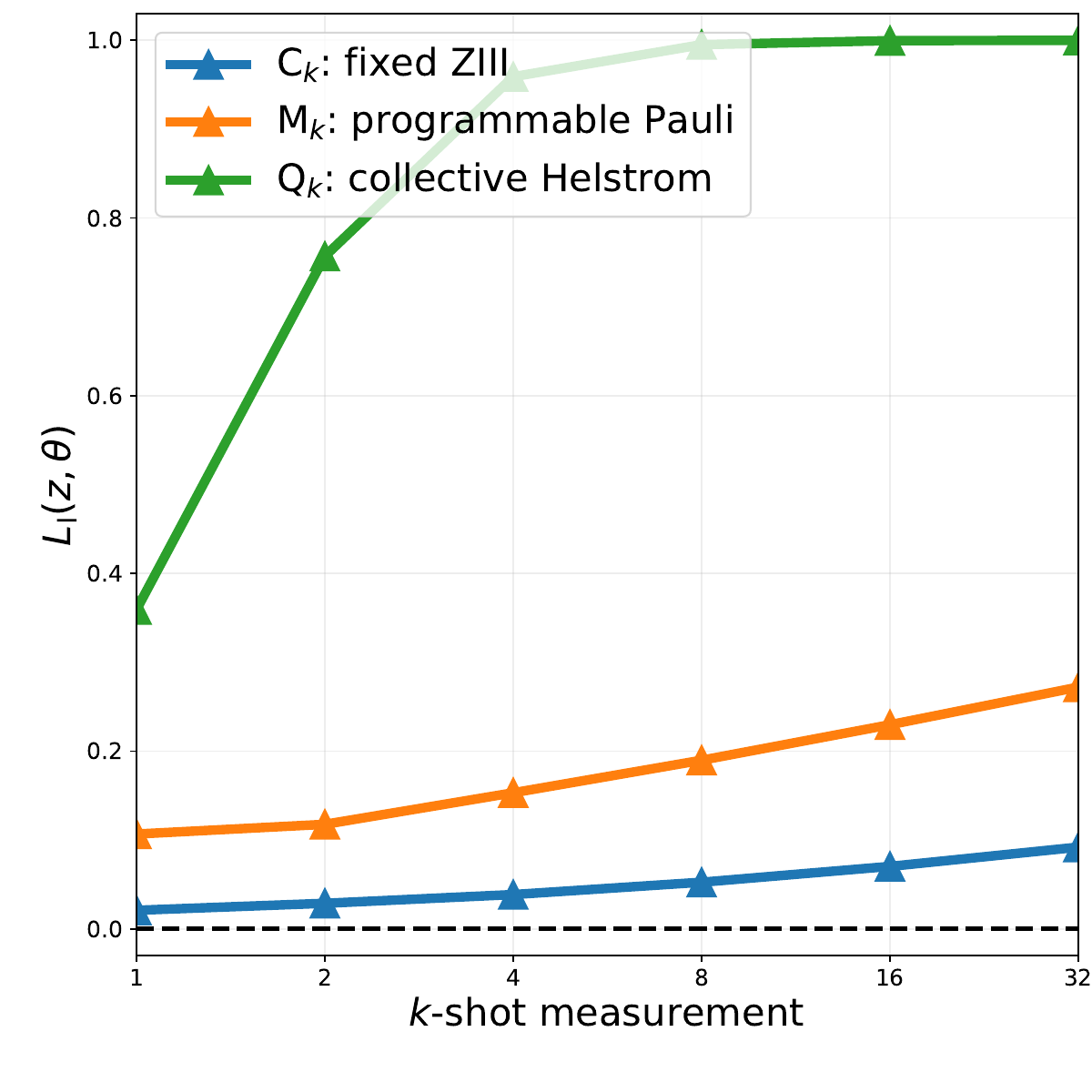}
        \caption{``3'' vs. ``8'', $K=64$.}
        \label{sfig:3v8_fix_surr}
    \end{subfigure}
    
    \begin{subfigure}[b]{.49\linewidth}
        \centering
        \includegraphics[width=\textwidth]{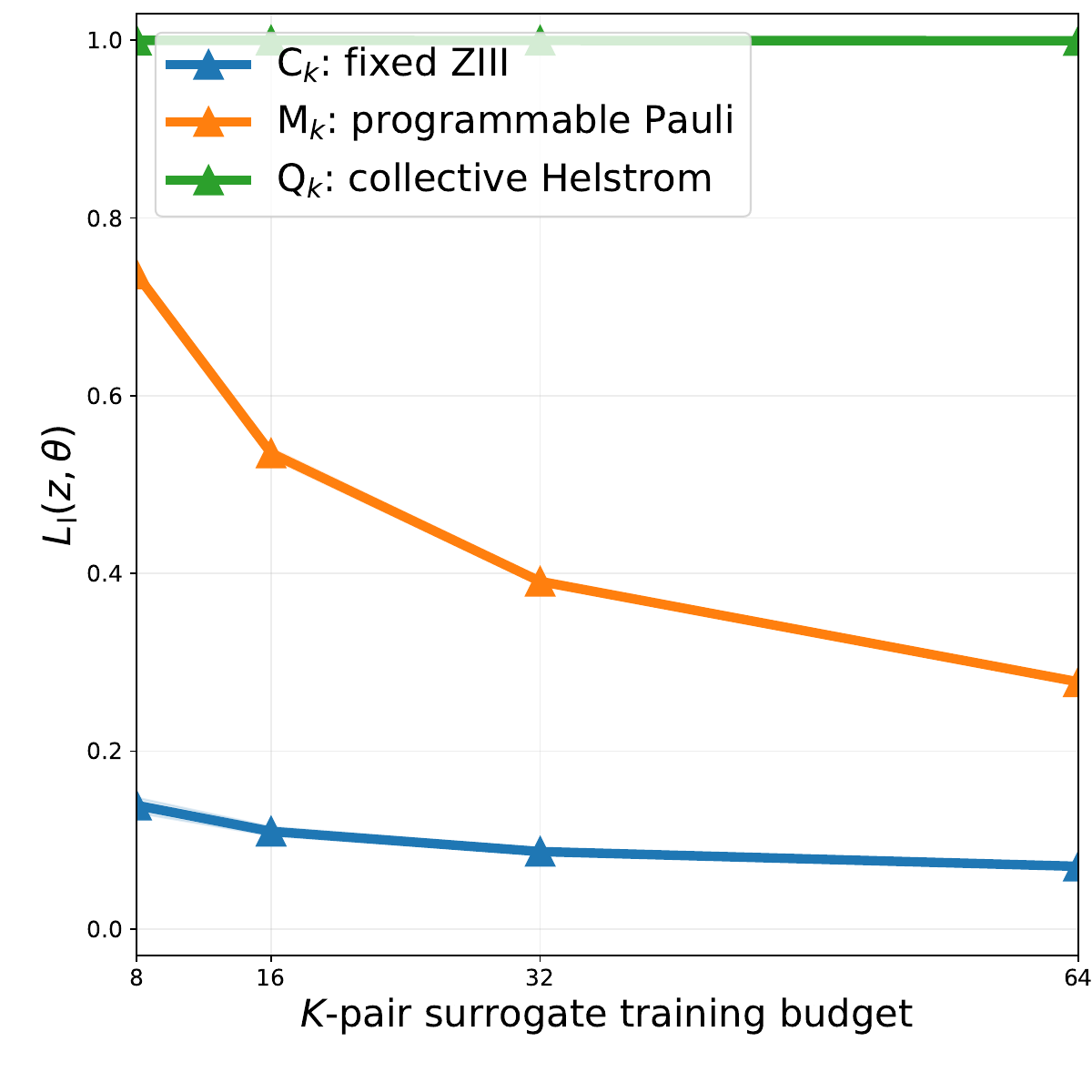}
        \caption{``0'' vs. ``1'', $k=32$.}
        \label{sfig:0v1_fix_shot}
    \end{subfigure}
    \begin{subfigure}[b]{.49\linewidth}
        \centering
        \includegraphics[width=\textwidth]{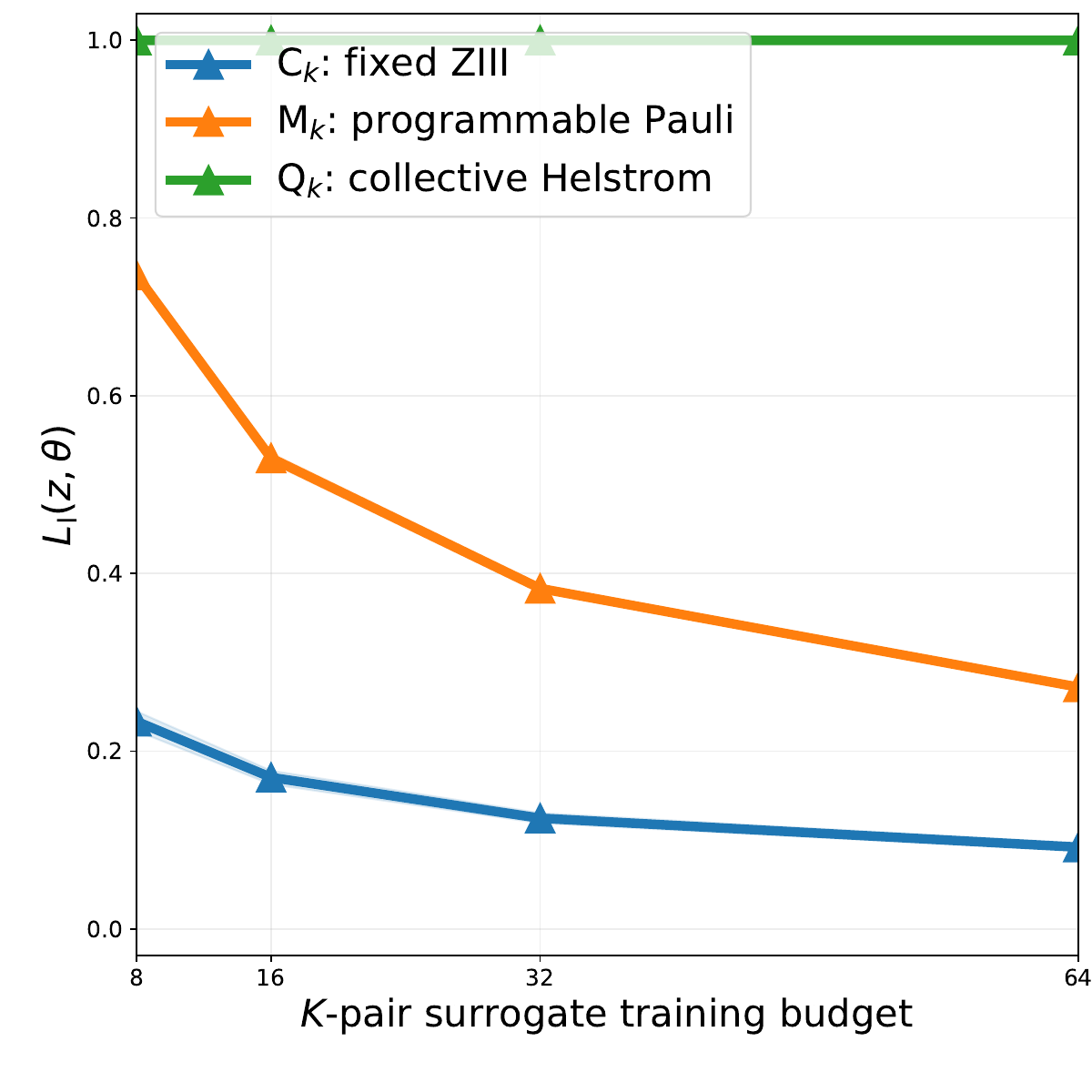}
        \caption{``3'' vs. ``8'', $k=32$.}
        \label{sfig:3v8_fix_shot}
    \end{subfigure}
    \caption{We study the membership inference advantage of $\Adv$ under different access regimes (Fig. \ref{sfig:0v1_fix_surr}-\ref{sfig:3v8_fix_shot}), measurement budget $k$ (Fig. \ref{sfig:0v1_fix_surr}, \ref{sfig:3v8_fix_surr}), and surrogate model training budget (Fig. \ref{sfig:0v1_fix_shot}, \ref{sfig:3v8_fix_shot}).}
    \label{fig:k-K-sweep}
\end{figure}

\begin{figure}[t]
    \centering
    \begin{subfigure}[b]{.49\linewidth}
        \centering
        \includegraphics[width=\textwidth]{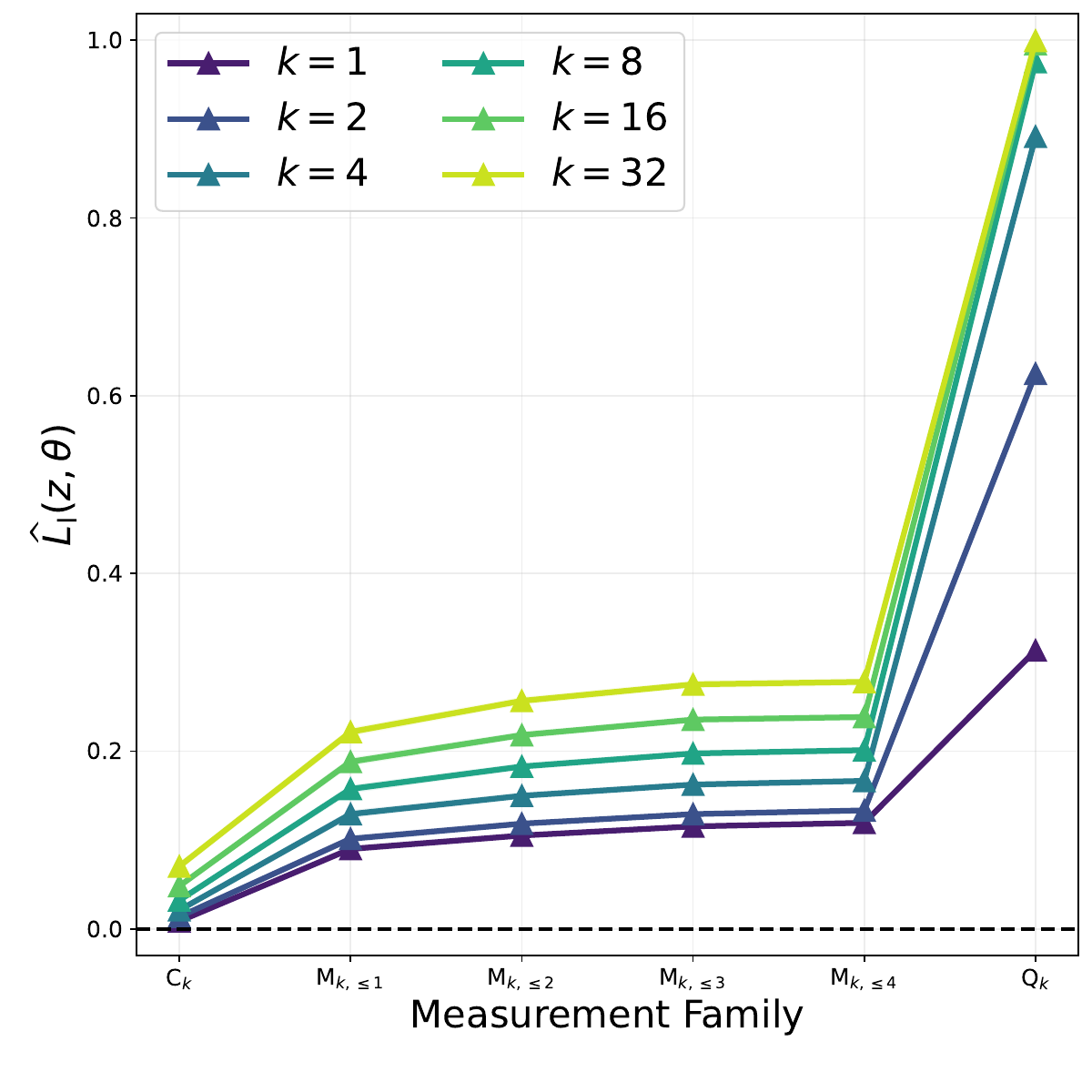}
        \caption{``0'' vs. ``1'', $K=64$.}
        \label{sfig:0v1_ladder}
    \end{subfigure}
    \begin{subfigure}[b]{.49\linewidth}
        \centering
        \includegraphics[width=\textwidth]{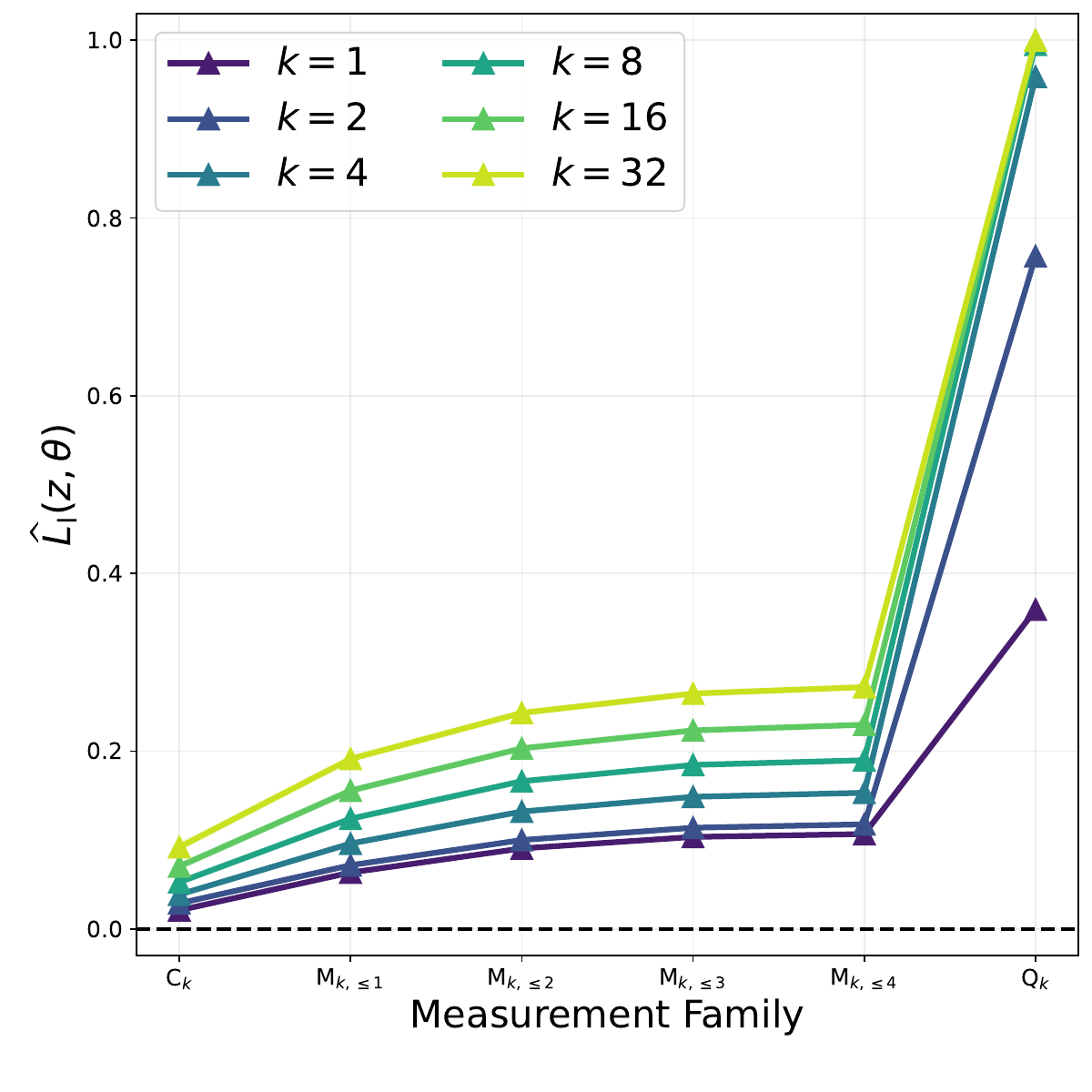}
        \caption{``3'' vs. ``8'', $K=64$.}
        \label{sfig:3v8_ladder}
    \end{subfigure}
    \caption{Membership inference advantage across POVMs with Pauli weight $\leq 4$ for $\mathsf{M}_k$.}
    \label{fig:access-ladder}
\end{figure}

\paragraph{Surrogate models} Following Carlini et al. \cite{Carlini2022MIA}, we construct the training set of the target as follows: we sample $512$ data points uniformly from the two task classes of MNIST as the global population of available data; for the training set of each surrogate model, we sample $|D|=|D^s|=256$ samples from the population. We train $K= \{ 8, 16, 32, 64 \}$ pairs of surrogate models for our experiments.

\paragraph{Membership inference attacks} As discussed in Sec. \ref{sec:theory}, we modify the features for MIAs used under each access regime: for $\mathsf{C}_k$, we fit the optimal subset of possible $ZIII$ counts; for $\mathsf{M}_k$, we evaluate all 255 possible Pauli string basis ${\langle X, Y, Z, I \rangle}^{\otimes 4}$ for $4$-qubit circuits for the supreme membership inference advantage; for $\mathsf{Q}_k$, we directly evaluate the empirical Helstrom measurement as an upper bound for the empirical performance of a quantum-capable $\Adv$ \cite{Helstrom1969Detection}.

\begin{figure*}[t]
    \centering
    \begin{subfigure}[b]{.25\linewidth}
        \centering
        \includegraphics[width=\textwidth]{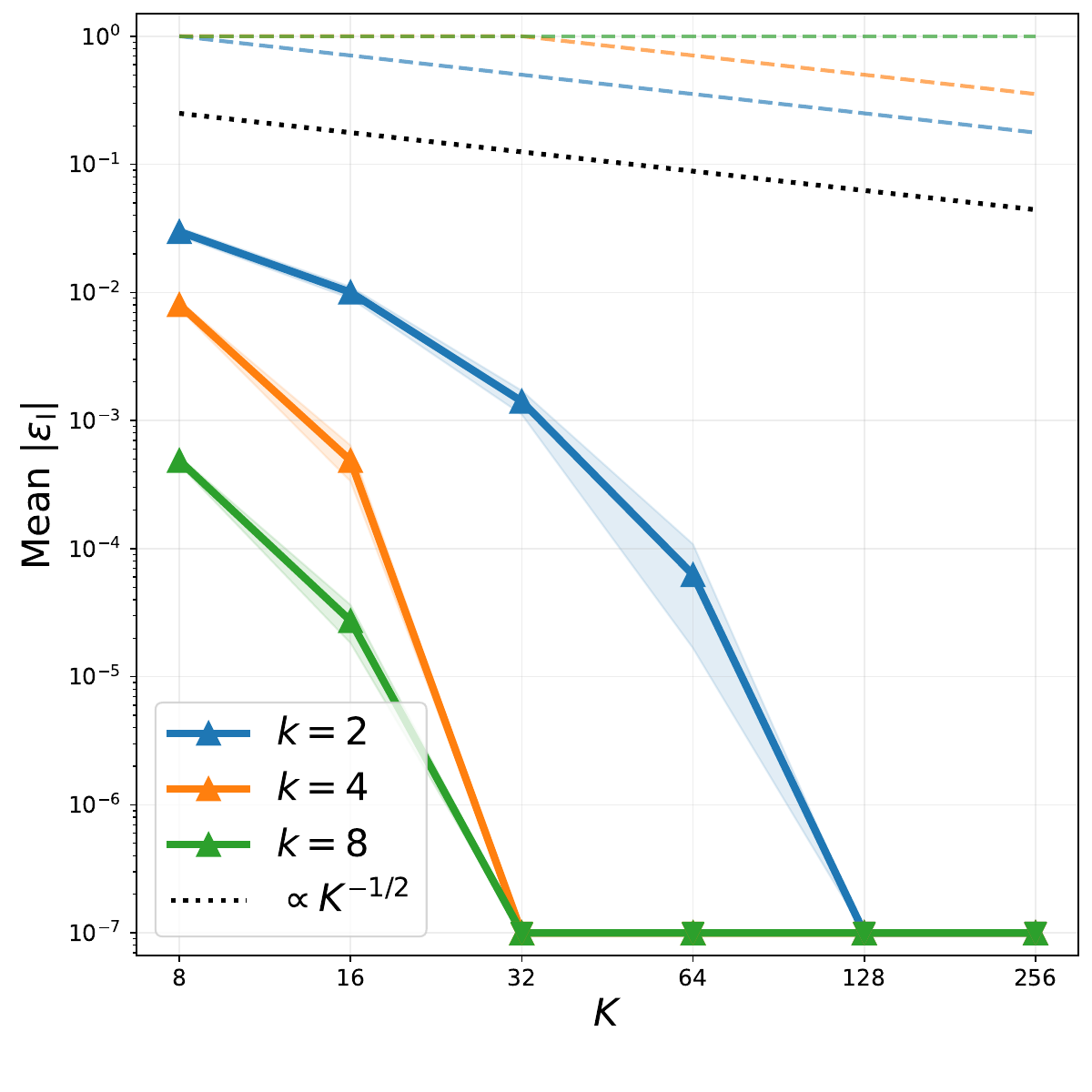}
        \caption{$\mathsf{C}_k$}
        \label{sfig:gap-C}
    \end{subfigure}
    \begin{subfigure}[b]{.25\linewidth}
        \centering
        \includegraphics[width=\textwidth]{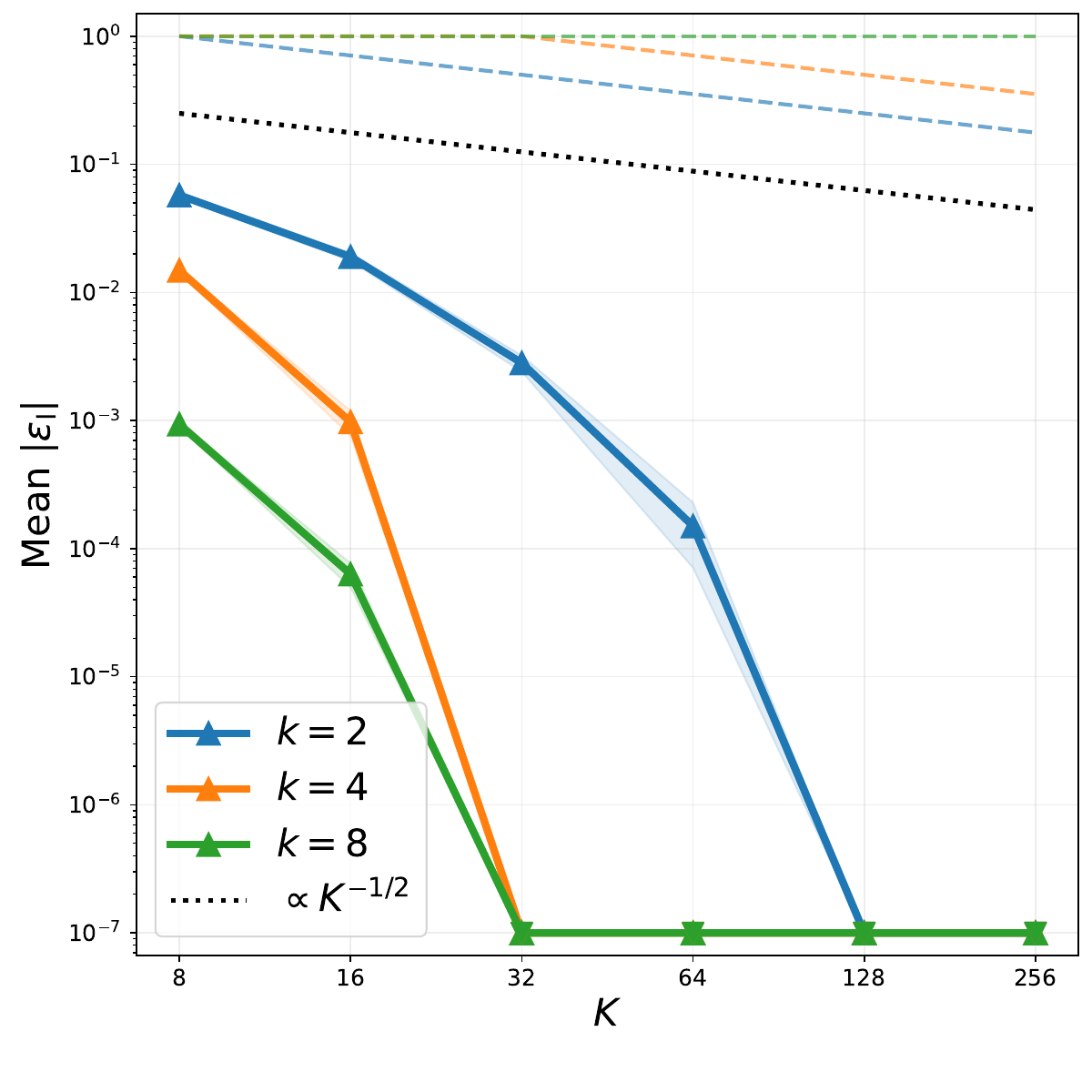}
        \caption{$\mathsf{M}_k$}
        \label{sfig:gap-M}
    \end{subfigure}
    \begin{subfigure}[b]{.25\linewidth}
        \centering
        \includegraphics[width=\textwidth]{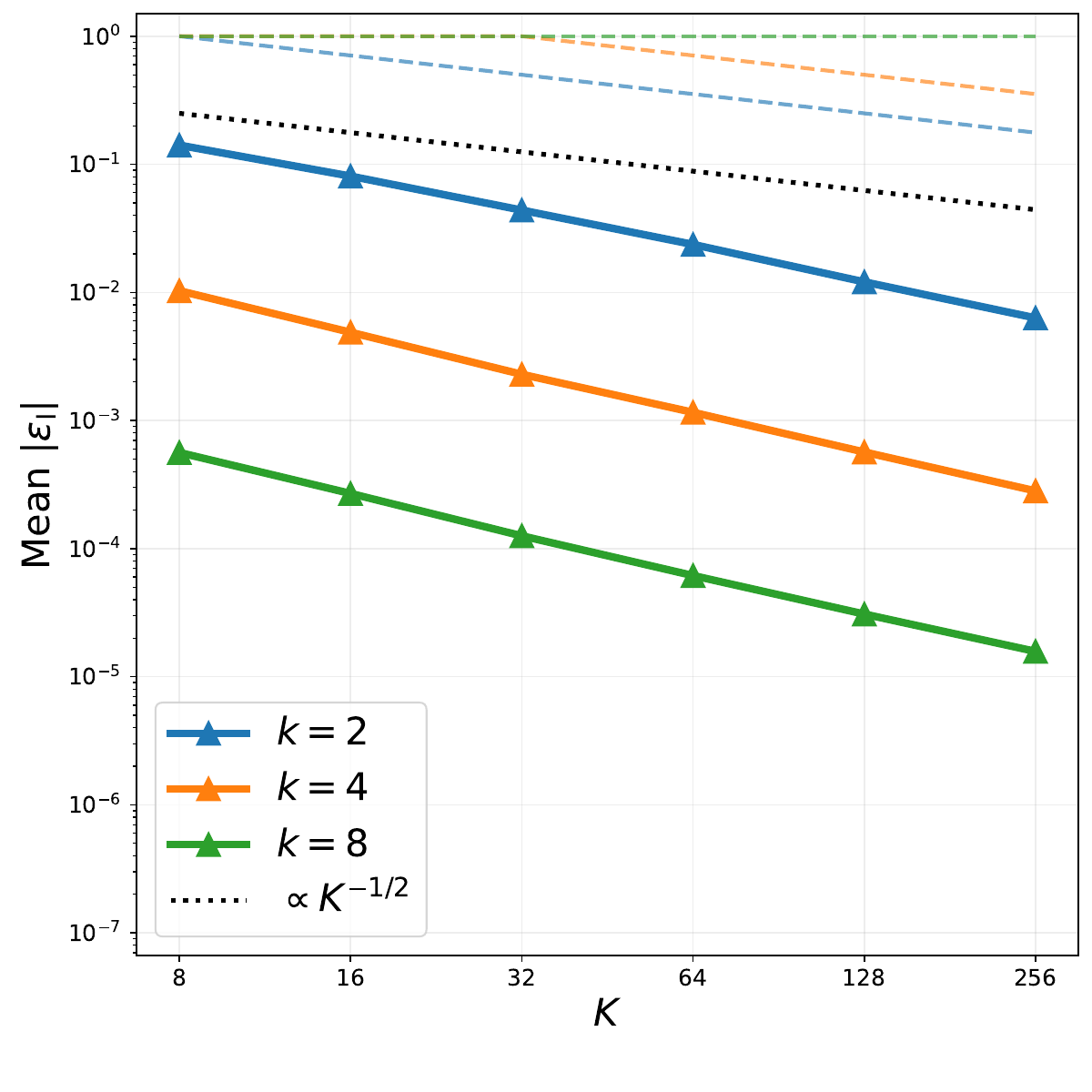}
        \caption{$\mathsf{Q}_k$}
        \label{sfig:gap-Q}
    \end{subfigure}
    \caption{Theoretical-empirical membership inference advantage gap across different access regimes.}
    \label{fig:gap}
\end{figure*}

\subsection{Results}\label{ssec:results}

\paragraph{Membership inference advantage gain} We plot the membership inference advantages across different access regimes $I \in \{ \mathsf{C}_k, \mathsf{P}_k, \mathsf{Q}_k \}$ with different shots and surrogate model training budgets. We plot the mean inference advantage for all target $z$ sampled from the data population (Sec. \ref{ssec:setup}). The results are demonstrated in Fig. \ref{fig:k-K-sweep}.

Fig. \ref{fig:k-K-sweep} strongly confirms our theoretical intuition: different access regimes impose a strict hierarchy formalized in Theorem \ref{th:hierachy}. With the same surrogate model budget, larger measurement budget $k$ increases the (empirical) adversary advantage $\widehat{L}_\mathsf{I}(z, \bm\theta)$ (Fig. \ref{sfig:0v1_fix_surr} and \ref{sfig:3v8_fix_surr}), as discussed in Observation \ref{obs:advantage-k}. Combined, these results present a clear answer to (\textbf{RQ2}): quantum access provides a demonstrable theoretical and empirical advantage for membership inference attacks.

\paragraph{Measurement selection} Beyond the 255 Pauli string basis considered for user/adversary-selected measurements used in Fig. \ref{fig:k-K-sweep}, we further provide a more fine-grained view of adversarial membership inference advantage under $\mathsf{M}_k$ by calculating $\widehat{L}_{\mathsf{M}}(z, {\bm\theta})$ under measurements with different \emph{Pauli weights}. For each measurement group with Pauli weights $\leq 1$ through $\leq 4$, we plot the supreme $\widehat{L}$ for each query budget $k$. We show the results in Fig. \ref{fig:access-ladder}, with access regimes $\mathsf{C}_k$ and $\mathsf{Q}_k$ as baselines.

From Fig. \ref{fig:access-ladder}, we establish that the specific selection of the measurement operator $M_u$ itself contributes to the membership inference advantage:  more complex measurements allow the user/adversary to preserve more information in $\rho_{\bm\theta}(x)$, leading to a strictly stronger inference advantage. At Pauli weight $\leq 4$, the adversary has a significant advantage over $\mathsf{C}_k$. As $\mathsf{M}_k$, which allows users to select measurements, represents the most eminent access regime \cite{Fu2024QuantumLeak,Devitt2016Cloud}, our result should be understood as a strong result on potential privacy leakage of QML models in the near future.

\paragraph{Theoretical-empirical leakage gap} Beyond empirical results on $\widehat{L}_\mathsf{I}(z, {\bm\theta})$, we further examine the theoretical-empirical leakage gap $\varepsilon_\mathsf{I}(z, {\bm\theta} \mid K)$ introduced in Theorem \ref{thm:finite} and Corollary \ref{cor:exp-bound}. We plot $\mathbb{E}_z \left[\varepsilon_\mathsf{I}(z, {\bm\theta} \mid K)\right]$ across different access regime $\mathsf{I}$, query budget $k$, and surrogate model pair $K$, the results are shown in Fig. \ref{fig:gap}.
Our result confirms Corollary \ref{cor:exp-bound} and Observation \ref{obs:gap}, where $k$ and $K$ jointly determine the gap between theoretical and empirical privacy leakages: a larger $K$ pushes the gap to $\to 0$, while a larger $k$ provides more information of the quantum state $\rho_{\bm\theta}(x)$.

We make two additional remarks on the result of Fig. \ref{fig:gap}. Firstly, we note that the bound given in Corollary \ref{cor:exp-bound} is not intended to be a tight bound, as demonstrated in Fig. \ref{sfig:gap-C} through \ref{sfig:gap-Q}, where $K^{-1/2}$ decays significantly slower than the actual observation. Secondly, different access regimes also determine how quickly the gap decays, where rich quantum information ensures the gap decays exponentially for $\mathsf{Q}_k$; adversaries can close the gap at a super-exponential rate for both $\mathsf{C}_k$ and $\mathsf{M}_k$, where existing works have demonstrated near-perfect MIAs under $\mathsf{C}_k$. We leave a tighter bound on this gap between theoretical and empirical privacy leakage to future works.


\section{Discussion and Future Work}\label{sec:dis}

\subsection{How Realistic are Quantum-Access Regimes?}\label{ssec:realistic}

As discussed in Sec. \ref{sec:intro} and \ref{sec:theory}, we again stress note that the three access regimes studied in this paper, $\mathsf{C}_k$, $\mathsf{M}_k$, and $\mathsf{Q}_k$ should be understood as abstractions of increasingly quantum-native service interfaces rather than all realized quantum computing services, in which:
\begin{itemize}
    \item $\mathsf{C}_k$ is well-aligned with QML-as-a-service, in which users remotely execute a model or circuit while receiving only classical measurement outcomes \cite{Garcia2021QaaS,Moguel2022QS}. 
    
    \item Remote access to programmable quantum circuits necessary for $\mathsf{M}_k$ has been discussed for cloud quantum-computing platforms by Devitt \cite{Devitt2016Cloud} and used in recent quantum devices \cite{Abughanem2025IBM}; Fu et al. \cite{Fu2024QuantumLeak} further studied remote query attacks against quantum cloud services.

    $\mathsf{M}_k$ therefore describes a plausible near-term regime in which the service provider retains the quantum computer and returns classical bits through regular channels, while still allowing user control over the measurement.

    \item The strongest $\mathsf{Q}_k$ regime remains more distant, but is motivated by the development of quantum networking (QN) and distributed quantum computing (DQC), which connect remote quantum processors through quantum communication \cite{Wehner2018Quantum} to exchange quantum resources for joint execution of algorithms \cite{Caleffi2024Survey}. Thus, it is necessary to consider privacy risks for QML models in a fully quantum-native environment, in which quantum states are directly exposed and measured only at the user end.
\end{itemize}
Our results do \emph{not} rely on any specific realizations of $\mathsf{M}_k$ or $\mathsf{Q}_k$, but rather aim to study the role of a \emph{classicalization} boundary in QC and QML in particular. We leave a more detailed analysis of MIA over the quantum internet for future work.

\subsection{Quantum Defenses in Quantum-Access Regimes}\label{ssec:def}

Beyond a characterization of privacy risks of QML models in the form of susceptibility to membership inference attacks, one research question naturally emerges: \emph{how can we defend QML models under these emergent quantum-native access regimes}, for which we propose several possible defense strategies.

Intuitively, as shown both theoretically (Theorem \ref{th:hierachy}) and empirically (Fig. \ref{fig:k-K-sweep} and \ref{fig:access-ladder}), $\mathsf{C}_k$ limits the available information to the adversary, and therefore itself constitutes a strong defense compared to $\mathsf{M}_k$ or $\mathsf{Q}_k$. A service provider can approach a bound on membership information leakage by restricting allowed user measurements $M_u$ for regime $\mathsf{M}_k$ (Fig. \ref{fig:access-ladder}) or by limiting the measurement budget $k$. However, this defense poses a strong \emph{privacy-utility trade-off} in which the performance of QML models is impacted, as discussed by Heredge et al. \cite{Heredge2025Character}.

An alternative line of defense is therefore to make the QML model \emph{itself} private on the training set $D$ through \emph{quantum differential privacy} (QDP) \cite{Zhou2017DPQC,Watkins2023QDP,Du2022QDPReg}. By definition, QDP-trained QML models bound the information contribution of individual training data to the trainable parameter $\bm\theta$ \cite{Dwork2006Calibrating}. Therefore, $\rho_{\bm\theta}(x)$ exposed by $\mathsf{Q}_k$, $M( \rho_{\bm\theta}(x) )$ exposed by $\mathsf{M}_k$ and $\mathsf{C}_k$, and the subsequent membership inference attacks are \emph{post-processing} of a differentially private $\bm\theta$, and cannot distinguish the membership status of any target sample $z$ beyond a $\varepsilon$- or $(\varepsilon, \delta)$-bound \cite{Dwork2006Calibrating,Zhou2017DPQC}. Moreover, existing works have demonstrated that quantum noise inherent in NISQ devices can itself serve as a source of privacy perturbation for QDP during the training of QML models \cite{Du2021QNoise,Watkins2023QDP,Ju2024Harness}. Therefore, there is a need for establishing a formal bound on the privacy leakage of differentially private QNNs under $\mathsf{C}_k$, $\mathsf{M}_k$, and $\mathsf{Q}_k$, respectively.

\section{Conclusion}\label{sec:conclusion}

\paragraph{Contributions} In this paper, we examine privacy leakages of quantum machine learning (QML) models in emergent quantum-native access regimes, in which users/adversaries possess quantum computing abilities and access to quantum information as outputs. We establish a hierarchy for adversarial advantage of membership inference attacks (MIAs) against the model under three regimes: classical outcome-only $\mathsf{C}$, programmable measurement $\mathsf{M}$, and quantum state access $\mathsf{Q}$, $\mathsf{C} \preceq \mathsf{M} \preceq \mathsf{Q}$. In addition, we prove a bound on the performance of empirical MIAs under limited query and surrogate-model training budgets. Our experimental results strongly confirm our theory and demonstrate how quantum information access and adversary computational resources jointly determine attack outcomes against a theoretically optimal leakage. Our work establishes a first characterization of QML privacy under quantum-native access

\paragraph{Limitations} Our work studies a small-scale (4-qubit, 24-parameter) QNN model under noiseless simulation, as demonstrated by us and confirmed by existing research on MIAs \cite{Carlini2022Onion}, MIA susceptibility are target instance- and model-dependent. Therefore, future research is much needed for a more generalized understanding of privacy risks of QMLs across a wider family of models (Def. \ref{def:qml}), wider definition of measurement operators (Table \ref{tab:access}), and larger-scale, noisy experiments (Sec. \ref{sec:eval}), which we aim to address in follow-up works.





\section*{Acknowledgments}
This work is funded by Cisco Research.

\bibliographystyle{IEEEtran}
\bibliography{IEEEabrv, bib}





\end{document}